\documentclass[aps,a4paper,preprintnumbers,11pt]{revtex4-2}

\usepackage{amsmath,amssymb,amsfonts,bbm,amsthm, changes, graphicx, dsfont, verbatim, appendix}
\usepackage[T1]{fontenc}  

\DeclareMathAlphabet{\mathcal}{OMS}{cmsy}{m}{n}

\usepackage[hyperindex, breaklinks]{hyperref}
\hypersetup{
     colorlinks=true,        
     citecolor=blue,    
     filecolor=blue,      		
     urlcolor=blue,           	
    runcolor=cyan,
}

\usepackage[left=2.5cm,right=2.5cm,top=2.5cm,bottom=2.5cm]{geometry}

\usepackage{orcidlink}
\usetikzlibrary{decorations.pathreplacing}

\theoremstyle{definition}
\newtheorem{Lemma}{Lemma}
\newtheorem{Definition}[Lemma]{Definition}
\newtheorem{Theorem}[Lemma]{Theorem}

\begin{document}

\newcommand{\B}{\textrm{Bott}}
\newcommand{\T}{\textrm{Tr}}

\newcommand{\PP}{P^\bot}
\newcommand{\I}{\mathds{1}}
\newcommand{\La}{\Lambda}

\newcommand{\eq}{Eq.}
\newcommand{\eqs}{Eqs.}
\newcommand{\cf}{\textit{cf. }}
\newcommand{\ie}{\textit{i.e. }}
\newcommand{\eg}{\textit{e.g. }}
\newcommand{\etal}{\emph{et al.}}
\def\i{\mathrm{i}}

\renewcommand{\thesection}{\arabic{section}}
\renewcommand{\thesubsection}{\thesection.\arabic{subsection}}
\renewcommand{\thesubsubsection}{\thesubsection.\arabic{subsubsection}}

\title{The Bott index of two unitary operators and the integer quantum Hall effect }

\author{Daniele Toniolo\,\orcidlink{0000-0003-2517-0770}}
\email{d.toniolo@ucl.ac.uk}

\affiliation{Department of Computer Science, University College London, United Kingdom}
\affiliation{Department of Physics and Astronomy, University College London, United Kingdom}

\date{\today}


\begin{abstract}
The Bott index of two unitary operators on an infinite dimensional Hilbert space is defined. Homotopic invariance with respect to  multiplicative unitary perturbations of the type identity plus trace class  and the ``logarithmic'' law for the index are proven. The index and its properties are then extended to the case of a pair of invertible operators. An application to the physics of two dimensional quantum systems proves that the index is equal to  the Chern number therefore showing that the transverse Hall conductance is an integer.  
\end{abstract}

\maketitle


\section{Introduction}

The Bott index of two unitary matrices arose to distinguish almost commuting unitary matrices and exactly commuting unitaries. Two equivalent formulations have been given in terms of a winding number and of a K-theoretic invariant in the works of Exel and Loring \cite{Loring_1988, Exel_Loring_1989, Exel_Loring_1991, Exel_1993}.

More recently Loring and Hastings have used the index in the context of a classification of topological phases of matter  \cite{Loring_Hastings_2010, Hastings_Loring_2010}, also with numerical applications \cite{Hastings_Loring_2011}.

The discovery of the quantization of the Hall conductance in disordered samples of two dimensional semiconductors subjected to a strong perpendicular magnetic field in 1980 \cite{von_Klitzing_1980} initiated the mathematical investigation of the phenomenon. Explanations in terms of the index of a Fredholm operator \cite{Bellissard_1988, Bellissard_1994} and of the index of two projections \cite{Avron_1990, Avron_Seiler_Simon_1994, Avron_Seiler_Simon_1994_2, Avron_les_houches_1994} have been provided. The integer quantum Hall effect is seen nowadays as an instance within the table of topological insulators first proposed by Kitaev \cite{Kitaev_2009}. A rigorous explanation of the table and of the associated phenomena of the bulk-boundary correspondence is given, employing K-theory, in the monograph \cite{Prodan_Schulz_Baldes}.
A new approach has recently emerged, particularly favorable for numeral estimations, where the topological index, dubbed spectral localizer, employs the system's Hamiltonian rather than the Fermi projection \cite{Loring_2015, Loring_Schulz_Baldes_2017, Loring_Schulz_Baldes_2019, Schulz_Baldes_2019, Schulz_Baldes_2021}. For a comparison among the Bott index and the spectral localizer also concerning their numerical implementation see \cite{Loring_2019}.   

A recent approach to topological phenomena that emphasises the many-body aspects is in the references \cite{Bachmann_2020,Bachmann_2021}. 

In the physics literature the Bott index has been employed to characterize several topological phenomena, besides the classification initiated in \cite{Hastings_Loring_2010}. To study time-dependent perturbations \cite{Refael_2015, Ge_Rigol_2017, Toniolo_2018}, time reversal invariant systems \cite{Loring_Hastings_2010} and in quasi periodic systems \cite{Huang_2018, Nielsen_2020, Yoshii_2021}. The former list is far from being exhaustive.

The Bott index is well defined and assumes integer values for unitaries defined on finite dimensional Hilbert spaces, this makes it particularly suitable for numerical computation of the index when the quantum system is defined for example on a 2-dimensional finite torus. On the contrary Fredholm operators with non vanishing index can only be defined on infinite dimensional Hilbert spaces.

The present work extends the definition of the Bott index to a pair of unitary operators on an infinite dimensional Hilbert space and prove features like homotopic invariance under multiplicative unitary perturbations of the type identity plus trace class and the ``logarithmic law''. For a review of these properties in the context of the Fredholm operators it is possible to refer to \cite{Avron_Seiler_Simon_1994_2,Murphy_1994,Lax_2002, Roe_2015}. 

The structure of the paper is as follows: In section \ref{definition} the Bott index of two unitary operators is defined. Section \ref{homotopy} discusses the classes of maps $ U(s) $, $ V(s) $ in the unitary operators that ensure homotopic invariance. Section \ref{logarithmic_law} establishes the ``logarithmic law'' namely $ \B(e^{iA} e^{iB},e^{iC}) =\B(e^{iA},e^{iC}) + \B(e^{iB},e^{iC}) $,
with $ A $, $ B $ and $ C $ self-adjoints satisfying certain trace class conditions. Section \ref{invertible} extends the definition of the Bott index to a pair of invertible operators, shortly discussing how the properties of the index carry over. Section \ref{2_dim_lattice} consider an explicit application to the physics of an infinite two dimensional lattice quantum system where the unitary operators $ U $ and $ V $ are in relation with the Fermi projection of a local Hamiltonian and with switch functions suitably defined on the lattice. In section \ref{Chern_number} the Bott index is shown to be equal to the Chern number of the Fermi projection $ P $. The quantization of the Hall conductance is also shown independently making use of a recent result of the literature \cite{Elgart_Fraas_2021}. The work concludes with a short discussion about connections of the theory of the Bott index with Fredholm index theory and the spectral localizer. Four appendices follow: \ref{appendix_C_constant} defines the operator $ \log $ and discusses its trace class properties,
\ref{vanishing_index} gives a sufficient condition for the vanishing of the Bott index. This result is presented as an appendix because more relevant to unitary matrices rather than operators. Appendix \ref{trace_class_pert} discusses unitary operators that differ by a trace class. The appendix \ref{logarithmic_law_proof} concludes with the proof of theorem \ref{theorem_mult_law} on the logarithmic law.


\section{The Bott index} \label{definition}
In what follows $ \mathcal{I}^1 $ denotes the ideal of trace class operators in the linear bounded operators $ \mathcal{L}(\mathcal{H}) $ on $ \mathcal{H} $. $ \| \cdot \|_1 $ is the associated trace norm, $ A \in \mathcal{I}_1 $ if and only if $ \| A \|_1 := \T \sqrt{A^*A} < \infty $.  Throughout the paper $ \mathcal{H} $ is supposed to be separable.

Following Kato the logarithm of a linear operator is defined by the Dunford integral, see lemma \ref{log} appendix \ref{appendix_C_constant}. The branch cut of the logarithm is fixed on the real negative axis.

The condition $ \|[U,V] \| <2 $ is equivalent to say that $ \{ -1 \} \notin \sigma \left( UVU^{-1}V^{-1}\right) $, in fact $ \| UVU^*V^* -1 \| = \|[U,V] \| $, therefore $  \log \left( UVU^*V^*\right) $ is well defined.

Given two unitary operators $ U $ and $ V $, if $ [U,V] \in \mathcal{I}^1 $ then $  \log \left( UVU^*V^*\right) \in \mathcal{I}^1 $. This follows from appendix \ref{appendix_C_constant} with $ S = UVU^{-1}V^{-1} $ and $ T = \I $.

\begin{Definition} \label{definition_bott}
Given two unitary operators $ U: \, \mathcal{H} \rightarrow \mathcal{H} $ and $ V: \, \mathcal{H} \rightarrow \mathcal{H}$, such that $ (U-\I)(V-\I) \in \mathcal{I}^1 $, $ (V-\I)(U-\I) \in \mathcal{I}^1 $ and $ \|[U,V] \| <2 $, the Bott index of $ U $ and $ V $  is defined as:
\begin{equation} \label{def_Bott}
 \B(U,V) := \frac{1}{2\pi i}\T \log \left( UVU^*V^*\right)
\end{equation}
$ \log $ denotes the principal logarithm, namely $ \log e^{i \theta } = i \theta $, with $ \theta \in (-\pi, \pi) $. 
\end{Definition} 

The conditions $ (U-\I)(V-\I) \in \mathcal{I}^1 $ and $ (V-\I)(U-\I) \in \mathcal{I}^1 $ implies, in particular, that $ [U,V] \in \mathcal{I}^1 $. This together with $ \|[U,V] \| <2 $ makes $ \log \left( UVU^*V^*\right) $ well defined and trace class.

\begin{Lemma}
 $ \B(U,V) $ as defined in eq. \eqref{def_Bott} is an integer. 
\end{Lemma}

\begin{proof}
 The work of Elgart and Fraas in \cite{Elgart_Fraas_2021} implies that with $ (U-\I)(V-\I) \in \mathcal{I}^1 $ and $ (V-\I)(U-\I) \in \mathcal{I}^1 $ then $ \det \left( UVU^*V^*\right) = 1 $, therefore denoting $ \{e^{i\theta} \} $, $ \theta \in (-\pi, \pi) $  the eigenvalues of $  UVU^*V^* $ it follows that:
 \begin{align} \label{integer}
  1=\prod_{j} e^{i\theta_j} = e^{i \sum_j \theta_j} =  e^{ \T \log \left( UVU^*V^*\right)} 
 \end{align}
This implies that $ \frac{1}{2\pi i}\T \log \left( UVU^*V^*\right) \in \mathds{Z} $.
In the first equality of eq. \eqref{integer} it is used the fact $ S \in \mathcal{I}^1 $ implies $ \det(\I + S) = \prod_j (1+ s_j) $, with $ \{s_j \} $ the eigenvalues of $ S $. In eq. \eqref{integer}, $ \T \log \left( UVU^*V^*\right) =i \sum_j \theta_j $ follows from Lidskii's theorem.
\end{proof}


\section{Homotopic invariance} \label{homotopy}

\begin{Theorem}
Given two maps $ U(s):[0,1] \rightarrow \mathcal{U}(\mathcal{H}) $ and $ V(s):[0,1] \rightarrow \mathcal{U}(\mathcal{H}) $ with $ U(0)=U $, $ V(0)=V$, such that $ \|[U(s),V(s)] \| <2 $ $ \forall s \in [0,1] $ and $ (U-\I)(V-\I) \in \mathcal{I}_1 $, $ (V-\I)(U-\I) \in \mathcal{I}_1 $.  If  $ \forall s \in [0,1] $, $ U(s)-U \in \mathcal{I}_1 $ and $ V(s)-V \in \mathcal{I}_1 $, then 
\begin{equation}
 \B(U(s),V(s))=\B(U,V)
\end{equation}
\end{Theorem}

\begin{proof}
Let us compute the derivative $ \partial_s \B(U(s),V(s)) $ assuming that $ \partial_s U(s) $ and $ \partial_s V(s) $ exist.
\begin{align}
 \partial_s \B(U(s),V(s)) &= \frac{1}{2 \pi i} \T \left[ \partial_s (U(s)V(s)U^*(s)V^*(s)) (V^*(s)U^*(s)V(s)U(s)) \right] \\
 &=\frac{1}{2 \pi i} \T [ (\partial_s U(s))U^*(s) + U(s)(\partial_s V(s))V^*(s)U^*(s) \nonumber \\ 
 & \hspace{1.2cm} + U(s)V(s)(\partial_sU(s)^*)U(s)V^*(s)U^*(s) \nonumber \\ & \hspace{1.2cm}  + U(s)V(s)U(s)^*(\partial_sV^*(s))V(s)U(s)V^*(s)U^*(s)   ] \\
 &= \frac{1}{2 \pi i} \T [ U^*(s)(\partial_s U(s)) + \partial_s V(s))V^*(s)+ V(s)(\partial_sU(s)^*)U(s)V^*(s) \nonumber \\ 
 & \hspace{1.2cm}  + V(s)U(s)^*(\partial_sV^*(s))V(s)U(s)V^*(s)  ] \\
 &= \frac{1}{2 \pi i} \T [ U^*(s)(\partial_s U(s)) + \partial_s V(s))V^*(s)- V(s)U^*(s) (\partial_sU(s))V^*(s) \nonumber \\ 
 & \hspace{1.2cm}  - V(s)U(s)^*V^*(s)(\partial_sV(s))U(s)V^*(s)  ] \\
 &= \frac{1}{2 \pi i} \T [ V^*(s) U^*(s)(\partial_s U(s)) V(s)  -U^*(s)(\partial_s U(s)) + V^*(s) (\partial_s V(s)) \nonumber \\ 
 & \hspace{1.2cm}  - U^*(s)V^*(s)(\partial_s V(s))U(s)  ]   \label{homotopic}
\end{align}
It occurs that $ U(s)-U \in \mathcal{I}_1 $, $ V(s)-V \in \mathcal{I}_1 $ if and only if $ U(s) = e^{isA}U $, $ V(s) = e^{isB}V$, with $ A $ and $ B $ trace class self-adjoint. This also implies that the maps $ U(s) $ and $ V(s) $ are continuous in the trace norm, as shown in appendix \ref{trace_class_pert}. Moreover $ U(s) $ and $ V(s) $ must satisfy $ \|[U(s),V(s)] \| <2 $ $ \forall s \in [0,1] $.  Equation \eqref{homotopic} then reads
\begin{align} \label{stability}
 \partial_s \B(U(s),V(s)) &= \frac{1}{2 \pi i} \T ( V^* e^{-isB}iA e^{isB}V  - i A - U^*e^{isA} iBe^{isA} U  + i B ) = 0
\end{align}
\end{proof}

Any unitary matrix $ U $ and $ V $ can be embedded in an infinite dimensional Hilbert space in the form identity plus trace class, this means that their Bott index is non vanishing if and only if the maps $ U(s) $, $ V(s) $ that connect them to the identity fail to satisfy  $ \|[U(s),V(s)] \| <2 $ $ \forall s \in [0,1] $. In this way we have recovered the homotopic invariance for the Bott index of matrices as already formulated by Exel and Loring \cite{Exel_Loring_1991,Hastings_Loring_2011}.

On the other hand in section \ref{Chern_number}  I will show an example of maps $ U(s) $ and $ V(s) $ such that $ U(0)=\I=V(0) $,  $ \|[U(s),V(s)] \| <2 $ $ \forall s \in [0,1] $ but with $ U(1)-\I $ and $ V(1)-\I $ not in the trace class. Along such a path the Bott index can vary. I will exploit this fact to recover the Chern number from the Bott index.


\section{Logarithmic law} \label{logarithmic_law}
Looking at the form of the Bott index, as given in the definition \ref{def_Bott}, the logarithmic law, that Fredholm operators satisfy, might not come as a surprise, but for matrices is actually wrong. An example easy to verify is given by the couple
\begin{align}
U=
\begin{pmatrix}
 0 & 0 & 1 \\
 1 & 0 & 0 \\
 0 & 1 & 0
\end{pmatrix}
\hspace{1cm} V= 
\begin{pmatrix}
 e^{\frac{2}{3}\pi i} & 0 & 0 \\
 0 & e^{-\frac{2}{3}\pi i} & 0 \\
 0 & 0 & 1
\end{pmatrix}
\end{align}
$ \B(U,V)=-1$, $ \B(U^2,V)=1$. Nevertheless it is possible to prove the following:
\begin{Theorem} \label{theorem_mult_law}
 Given $ A $, $ B $ and $ C $ self-adjoint operators such that $ [A,B] \in \mathcal{I}_1 $, $ [C,B] \in \mathcal{I}_1 $, $ \|[e^{itA},e^{itB}] \| <1 $ and $ \| [e^{itC},e^{itB}] \| < 1 $, $ \forall t \in [0,1] $, it holds:
\begin{equation}
   \T \log \left( e^{iA}e^{iB}e^{iC}e^{-iB}e^{-iA}e^{-iC} \right) = \T \log \left( e^{iA}e^{iB}e^{-iA}e^{-iB} \right) + \T \log \left( e^{iB}e^{iC}e^{-iB}e^{-iC} \right)
\end{equation}
With the further hypothesis $ (e^{iA}-\I)(e^{iB}-\I) \in \mathcal{I}_1 $ and $ (e^{iB}-\I)((e^{iC}-\I) \in \mathcal{I}_1 $, it is
\begin{equation}
 \B(e^{iA}e^{iB},e^{iC})=\B(e^{iA},e^{iC}) + \B(e^{iB},e^{iC})
\end{equation}
\end{Theorem}
The proof follows the same ideas employed in the proof on the equivalence Bott-Chern therefore it is postponed after that in appendix \ref{logarithmic_law_proof}.

It is an easy consequence of the Baker-Campbell-Hausdorff formula, see for example \cite{Miller, Hall} for reference, that $ \T \log $ is additive in a neighbor of the identity, we are not making such an assumption here.  

The class of unitary operators that can be written in the form $ e^{A} $, with $ A $ skew self-adjoint, is rather large. Any unitary operator $ W $ with a spectral gap, namely with spectrum strictly contained in the complex unit circle, is such that there exist skew self-adjoint operators $ A_1,\ldots, A_n$ with $ W = \prod_j e^{A_j} $, proposition 2.1.6 of \cite{Rordam_2000}. Moreover the spectral gap condition implies that $ W $ is homotopic to the identity within the unitary operators (there exist a norm continuous map with initial point the identity and final point $ W $).

The example with the unitary matrices $ U $ and $ V $ presented at the beginning of this section does not fit into theorem  \ref{theorem_mult_law}. Any pair of unitary matrix can be written as $e^{iS}$, $e^{iT}$ with $ S $ and $ T $ Hermitean, this follows from the spectral theorem and equation \eqref{exp_pro} in appendix \ref{trace_class_pert}. If $\| [e^{itS},e^{itT}] \| <2 $ for all $ t \in [0,1] $ then $ \B \left( e^{itS},e^{itT} \right) $ is vanishing because of homotopic invariance. Since $ \B(U,V) =-1 $ the condition $\| [e^{itS},e^{itT}] \| <2 $ cannot be met. We have seen in section \ref{homotopic} that for two unitary operators, $\| [e^{itA},e^{itB}] \| <2 $  does not imply that $ \B(e^{itA},e^{itB}) $ is vanishing.


\section{The Bott index of two invertible operators} \label{invertible}
The Bott index of two invertible operators, $ \rho $ and $ \eta $ can be defined similarly as done for unitary operators. To ensure that $  \log \left( \rho \eta \rho^{-1} \eta^{-1} \right) $ is well defined, having chosen the branch cut of the logarithm on the real negative axis,  we need that $ \sigma( \rho \eta \rho^{-1} \eta^{-1} ) $ does not contain any real negative value $ \sigma( \rho \eta \rho^{-1} \eta^{-1} ) \bigcap \mathds{R^-} =\emptyset $.  The trace class conditions that ensure  $  \log \left( \rho \eta \rho^{-1} \eta^{-1} \right) \in \mathcal{I}_1 $ and $ \B( \rho,\eta) \in \mathds{Z} $ are similarly to the unitary case, see \cite{Elgart_Fraas_2021}, $ ( \rho - \I)(\eta-\I) \in \mathcal{I}_1 $, $ (\eta-\I)( \rho - \I) \in \mathcal{I}_1 $. This implies $ \rho \eta \rho^{-1} \eta^{-1} = \I + \textrm{trace class}$. Denoting with $ \lambda_j = |\lambda_j|e^{i\theta_j} $, $ \theta_j \in (-\pi, \pi) $, the set of eigenvalues of $ \rho \eta \rho^{-1} \eta^{-1} $ we get:
\begin{equation} \label{prod_eign}
 1= \det \left( \rho \eta \rho^{-1} \eta^{-1} \right) = \prod_j \lambda_j  = \prod_j |\lambda_j| e^{i\theta_j} = \prod_j e^{i\theta_j}  \Rightarrow \frac{1}{2 \pi}\sum_j \theta_j \in \mathds{Z}  
\end{equation}
In equation \ref{prod_eign} it has been used: $ 1 =  \prod_j \lambda_j = |\prod_j \lambda_j| = \prod_j |\lambda_j| $, implying  $ 0 = \sum_j \log |\lambda_j| $, then
\begin{equation} \label{unitary_part}
 \B( \rho,\eta) := \frac{1}{2 \pi i} \T \log \left( \rho \eta \rho^{-1} \eta^{-1} \right) = \frac{1}{2 \pi i} \sum_j \log \lambda_j =
 \frac{1}{2 \pi i} \sum_j \left( \log |\lambda_j| + i \theta_j \right) = \frac{1}{2 \pi} \sum_j  \theta_j \in  \mathds{Z} 
\end{equation}


The polar decomposition of $ \rho $ is $ \rho = U_\rho |\rho| $, with $ U_\rho := \rho |\rho|^{-1} $ unitary. If $ \rho $ is normal then $ U_\rho $ and $ |\rho| $ commute, in general this is not the case. 

Let us show that the positive operators $ |\rho| $ and $|\eta| $ crucially affect $ \B(\rho,\eta) $.
 We start considering that the proof of the homotopic invariance of the Bott index of unitary operators carries unchanged to the case of invertible operators. Let $ \eta $ be unitary and consider the homotopy $ \rho(s)=U_\rho|\rho|^s $ with $ s\in [0,1] $. $ |\rho| > 0 $ therefore $ |\rho|^s $ is well defined. Equation \eqref{homotopic} then reads:
\begin{align}
\partial_s \B(\rho(s),\eta) = \frac{1}{2 \pi i} \T ( \eta^{-1} s|\rho|^{-1}  \eta  - s|\rho|^{-1} )
\end{align}
If the trace above vanishes then $ \B(\rho,\eta)=\B(U_\rho,\eta) $.

The following lemma provides a generalization of Lemma \ref{stability} where the unitary operators $ U $ and $ V $ are perturbed by invertible operators of the type identity plus trace class. 
 
\begin{Lemma}
Given the invertible operators $ \rho $ and $ \eta $ such that $ \sigma( \rho \eta \rho^{-1} \eta^{-1} ) \bigcap \mathds{R^-} =\emptyset $, $ ( \rho - \I)(\eta-\I) \in \mathcal{I}_1 $, $ (\eta-\I)(\rho - \I) \in \mathcal{I}_1 $. Let us consider $ \rho $ and $ \eta $ of the following form $\rho=Ue^{K+iR} $ and $ \eta=Ve^{S+iT}$, with $ U $ and $ V $ unitary, $ K, R, S ,T $ self-adjoint trace class. $ U $ and $ V $ are such $ \B(U,V) $ is a well defined integer according to \ref{definition_bott}. It holds:
 \begin{equation}
  \B(\rho,\eta)=\B(U,V)
 \end{equation}
\end{Lemma}

\begin{proof}
 As noted above the homotopic invariance carries to the case of invertible operators. Therefore we consider the maps $ \rho(a)=Ue^{a(K+iR)} $ and $ \eta(a)=Ve^{a(S+iT)} $, with $ a \in [0,1] $. Then $ \partial_a \rho(a) $ and $ \partial_a \eta(a) $ are trace class implying, according to equation \eqref{homotopic}, that $ \partial_a \B(\rho(a),\eta(a))=0$. 
\end{proof}

\section{Application: integer quantum Hall effect} \label{2_dim_lattice}

In this section we consider an application of the general formalism just developed to the physics of quantum Hall systems. Let us consider a bounded self-adjoint $ H $ (a hamiltonian) acting on the Hilbert space of square integrable sequences on the 2-dimensional square lattice $ l_2(\mathbb{Z}^2) $. Let the hamiltonian $ H $ to represent an insulating phase of matter without translation invariance but with a spectral gap. It is well known that the transverse Hall conductance is quantized, due to the stability of the index of the underlining Fredholm operator, in strongly disordered systems without a spectral gap but with a mobility gap, namely when states at the Fermi energy are localized. This has been realized already in \cite{Bellissard_1994} see also \cite{Aizenman_Graf_1998, Elgart_Graf_Shenker, Prodan_Hughes_Bernevig_2010, Prodan_review_2011}. Nevertheless the requirement of a spectral gap makes the presentation straightforward. 

I briefly present the formalism of switch functions and state two crucial results that I freely use: given a local operator  
$ A $, $ [\Lambda_x, A ] $ is $x$-localized and $ [\Lambda_y, A ] $ is $y$-localized. The product of any bounded operator with a $x$-localized operator is $ x$-localized, similarly with a $y$-localized operator. The product of an $x$-localized operator with a $y$-localized operator is a trace class operator. The formalism of switch functions have been used for example in  \cite{Avron_Seiler_Simon_1994}. Very readable expositions are those of \cite{Graf_Tauber_2018, Shapiro_Tauber_2019}, in particular the last reference extends the formalism to quasi-local operators.

A basis of $ l_2(\mathbb{Z}^2) $ is the set 
\begin{equation}
 \delta_n(x) = 
 \left\{ \begin{array}{ll}
    1 \ \ \mbox{ if } x = n \in \mathbb{Z}^2 \\
    0 \ \ \mbox{ otherwise } 
  \end{array} \right.\ 
\end{equation}
The Hamiltonian $ H $ is assumed to be short range, this is common for tight binding models, namely $ \exists R > 0 $ finite, such that for every pair of points on the lattice $ n=(n_x,n_y) $, $ m=(m_x,m_y) $, if $ \textrm{dist}(n,m) > R $ then $ \langle \delta_m, H \delta_n \rangle  = 0 $. The assumption of $ H $ short range can be relaxed, and a local $ H $ can be considered, but again this choice has been made to make the presentation simpler. Let us define $ P $ the Fermi projection, $ P := \chi_{(-\infty, \, \mu)}(H) $, the chemical potential $ \mu $ is assumed to be within the spectral gap. Let us also define $ \Lambda_x(x) $ with $ x \in \mathds{R} $ as a switch function: there exist $ L_l < 0 $ and $ L_r > 0 $ such that:
\begin{equation}
 \Lambda_x(x) = 
 \left\{ \begin{array}{ll}
    1 \ \ \mbox{ if } x > L_r \\
    0 \ \ \mbox{ if } x < L_l 
  \end{array} \right.\ 
\end{equation}
with $ L_l \le x \le L_r $, $ \Lambda_x(x) $ is taken increasing but in general not continuous. Given the position operator $ X $, with the functional calculus we define the operator $ \Lambda_x(X) $, namely: $ \Lambda_x(X)\delta_{(n_x,n_y)} = \Lambda_x(n_x)\delta_{(n_x,n_y)} $.
$ \Lambda_y(y) $ is analogously defined. In what follows $ \Lambda_x(X) $ will be denoted $\Lambda_x$, similarly for $ \Lambda_y(Y)$,  denoted $\Lambda_y$.

We consider the two unitary operators $ e^{2\pi i P\Lambda_xP} $ and $ e^{2\pi i P\Lambda_yP} $ acting on $ l_2(\mathbb{Z}^2) $, $ P $ is the Fermi projection defined above. These operators where considered by Kitaev in the appendix C of \cite{Kitaev_2006}, and more recently in \cite{Shapiro_2020, Bols_2021}.

We need to verify that  $ e^{2\pi i P\Lambda_xP} $ and $ e^{2\pi i P\Lambda_yP} $ comply with the definition of the Bott index.
Let us show that with properly chosen switches $ \La_x $ and $ \La_y $, it holds $ \| \left[ e^{2\pi i P\Lambda_xP}, e^{2\pi i P\Lambda_yP} \right] \| < 2 $.

We start bounding $ \| [\La_x,H] \| $. The Holmgren bound for the norm of a bounded operator $ A $ is:
\begin{equation}
 \| A \| \le \max \left( \sup_{m\in \mathbb{Z}^2} \sum_{n \in \mathbb{Z}^2} | \langle \delta_m , A \delta_n \rangle |, m \leftrightarrow n \right)  
\end{equation}
A proof of this bound can be found in chapter 16 of \cite{Lax_2002}.
\begin{align}
  \| [\La_x,H] \| \le \max \left( \sup_{m\in \mathbb{Z}^2} \sum_{n \in \mathbb{Z}^2} | \langle \delta_m , [\La_x,H] \delta_n \rangle |, m \leftrightarrow n \right) 
\end{align}
We notice that 
\begin{equation}
 \langle \delta_m , [\La_x,H] \delta_n \rangle = \langle \delta_m , (\La_xH - H \La_x) \delta_n \rangle = (\La_x(m_x)-\La_x(n_x))\langle \delta_m , H \delta_n \rangle
\end{equation}
Therefore
\begin{align} \label{finite_range}
  \| [\La_x,H] \| \le \max \left( \sup_{m\in \mathbb{Z}^2} \sum_{\textrm{dist}(n,m)\le R }  |\La_x(m_x)-\La_x(n_x)| |\langle \delta_m ,H \delta_n \rangle |, m \leftrightarrow n \right) 
\end{align}
In eq. \eqref{finite_range} we took into account that given a fixed point $ m \in \mathbb{Z}^2 $ only the points of the lattice within the range $ R $ can contribute to $ \langle \delta_m ,H \delta_n \rangle $.

The variation of the switch $ \La_x $, $ |\La_x(m_x)-\La_x(n_x)| $ is either vanishing, when both $ m_x$ and $ n_x $ belong to a region where $ \Lambda_x $ is constant or can be made small for example choosing $ \La_x $ linearly increasing from $ 0 $ to $ 1 $ over a length $ L = L_r-L_l$ centered around $ x=0 $ with $ L \gg R $.  
This implies that 
\begin{equation}
\| [\La_x,H] \| \le \mathcal{O}\left(\frac{R}{L}\|H\| \right) 
\end{equation}
Using the contour integral representation of the Fermi projection $ P $ it is easy to see that:
\begin{equation}
\| [\La_x,P] \| \le \mathcal{O}\left(\frac{R}{L} \frac{\|H\|}{\Delta E} \right) 
\end{equation}
$ \Delta E $ is the energy gap where the chemical potential $ \mu $ is placed. $ L $ is chosen such that $ \left(\frac{R}{L} \frac{\|H\|}{\Delta E} \right) \ll 1 $.
The above bounds are the same as those established in the case of a finite dimensional Hilbert space by Hastings and Loring \cite{Hastings_Loring_2010,Hastings_Loring_2011,Loring_Hastings_2010}. In that case $ L $ was representing the diameter of the sphere (or the ``linear'' dimension of the torus) where the tight binding Hamiltonian was realized.
We can now bound $ \| \left[ e^{2\pi i P\Lambda_xP}, e^{2\pi i P\Lambda_yP} \right] \|  $.
\begin{align}
 & \| \left[ e^{2\pi i P\Lambda_xP}, e^{2\pi i P\Lambda_yP} \right] \| \\
 & = \|  e^{2\pi i P\Lambda_xP} e^{2\pi i P\Lambda_yP}e^{-2\pi i P\Lambda_xP} -  e^{2\pi i P\Lambda_yP} \| \\
 & = \| \int_0^1 e^{2\pi i t P\Lambda_xP}\left[2\pi i  P\Lambda_xP, e^{2\pi i P\Lambda_yP}\right]e^{-2\pi i t P\Lambda_xP} dt \| \\
 & \le \| [2\pi i  P\Lambda_xP, e^{2\pi i P\Lambda_yP}] \| \\
 & = \| \int_0^1 e^{-2\pi i s P\Lambda_yP}\left[2\pi i  P\Lambda_xP, 2\pi i P\Lambda_yP\right]e^{2\pi i s P\Lambda_yP} ds \| \\
 & \le \| \left[2\pi i  P\Lambda_xP, 2\pi i P\Lambda_yP\right] \| = 4 \pi^2 \| \left[ P\Lambda_xP, P\Lambda_yP\right] \|\\
 & \le 4 \pi^2 \left( \| \left[ \Lambda_x, P\right] \| + \| \left[ \Lambda_y, P\right] \| \right) \ll 1 \label{final_bound}
\end{align}
I suggest and sketch a different approach to  bound $ \| \left[ e^{2\pi i P\Lambda_xP}, e^{2\pi i P\Lambda_yP} \right] \| $.
Consider $ f $ and $ g $ such that $ k\hat f \, \textrm{and} \, k\hat g \in L^1(\mathds{R}) $ then it holds:  $ \| [f(A),g(B)] \| \le \|k\hat f\|_1 \|k\hat g\|_1 \|[ A,B ] \| $, a proof can be inferred by theorem 5.8.8 of the book \cite{Simon_Operator_Theory}.
Let us look at the operator $ e^{2 \pi i P \La_x P } $, it holds $ \|P \La_x P\| \le 1 $. We can imagine to extend the function $ e^{2 \pi i x } $ on the interval $ (-\infty , -1) \cup (1, \infty) $ to a function $ f $, a Gaussian for example, with smooth enough junctions; in doing so the operator $ e^{2\pi i P\Lambda_xP} $ is unaffected. We do the same for the $y$-partner. Since
\begin{equation}
 \| [ P\Lambda_xP, P\Lambda_yP] \| = \| [ P\left(\Lambda_x-\frac{1}{2}\right)P, P\left(\Lambda_y-\frac{1}{2}\right)P] \| \le \frac{1}{2}
\end{equation}
to ensure $ \| \left[ e^{2\pi i P\Lambda_xP}, e^{2\pi i P\Lambda_yP} \right] \| < 2$,  we need that the function $ f $ described above satisfies $ (\|k\hat f\|_1)^2 < 4 $. 

This approach would also allow to avoid the explicit use of the spectral gap $ \Delta E $ opening up to the possibility of working in a mobility gap regime where the Fermi projection is defined through the Borel functional calculus rather than the continuous one. This approach would moreover extend the class of the Hamiltonian from finite range to local, and make the formulation of the Bott index, given in theorem \ref{Bott_physics} below, completely independent from the shape of the switch functions.
I do not have a clean proof of $ (\|k\hat f\|_1)^2 < 4 $ but only some numerical hints.

We now move on to discuss the trace class properties of $ \left( e^{2\pi i P\Lambda_xP} -\I \right) \left( e^{2\pi i P\Lambda_yP} -\I \right) $. A result of \cite{Shapiro_2020}, proposition 4.4, states that if a local operator $ A $ is such that $ A^2 - A $ or $ A^2 + A $ is $x$-confined then $ \left( e^{-2\pi i A } -\I \right) $ is $x$-confined.  To prove that $ \left( e^{2\pi i P\Lambda_xP} -\I \right) $ is $x$-confined, we simply verify that $ (P\Lambda_xP)^2 - P\Lambda_xP $ is $x$-confined. Denoting $ \PP := \I -P $, we have:
\begin{align}
 & (P\Lambda_xP)^2 - P\Lambda_xP = P\Lambda_xP\Lambda_xP - P\Lambda_xP = P\Lambda_x(\I -\PP)\Lambda_xP - P\Lambda_xP \\
 &= P(\La_x^2 -\La_x)P -P\Lambda_x\PP\Lambda_xP = P(\La_x^2 -\La_x)P - [P,\Lambda_x]\PP[\Lambda_x,P] 
\end{align} 
We observe that the function $ \La_x(x)^2 -\La_x(x)$ is compactly supported, namely vanishing for $ |x| \ge \max\{L_r,|L_l|\} $, this implies, with a slight abuse of language, that $ P(\La_x^2 -\La_x)P $ is $x$-confined. $ [P,\Lambda_x] $ is $x$-confined according to what stated above.

Repeating the same argument for the $y$-case we have that $ \left( e^{2\pi i P\Lambda_yP} -\I \right) $ is $y$-confined, making $ \left( e^{2\pi i P\Lambda_xP} -\I \right) \left( e^{2\pi i P\Lambda_yP} -\I \right) $ and $ \left( e^{2\pi i P\Lambda_yP} -\I \right) \left( e^{2\pi i P\Lambda_xP} -\I \right) $ both trace class.

\begin{Theorem} \label{Bott_physics}
Given a range-$ R $ Hamiltonian, $ H $, on $ l_2(\mathbb{Z}^2) $ with a spectral gap, denoting $ P=\chi_{(-\infty,\mu)}(H) $, with $ \mu $ in the spectral gap, the Fermi projection, and $ \La_x(x) $, $ \La_y(y) $ two switch functions increasing from zero to one over a length $ L \gg R $, it follows that
\begin{equation}
 \B(e^{2\pi i P\Lambda_xP}, e^{2\pi i P\Lambda_yP}) := \frac{1}{2\pi i} \textrm{Tr} \log \left(e^{2\pi i P\Lambda_xP} e^{2\pi i P\Lambda_yP} e^{-2\pi i P\Lambda_xP} e^{-2\pi i P\Lambda_yP} \right)
\end{equation}
is a well defined integer.
\end{Theorem}
\begin{proof}
 Given above.
\end{proof}


\section{Equivalence with the Chern number} \label{Chern_number}

We will now show that $ \B(e^{2\pi i P\Lambda_xP}, e^{2\pi i P\Lambda_yP}) $ equals the transverse Hall conductance $ \sigma_{\textrm{Hall}}=2\pi i \T \left[P\La_xP,P\La_yP\right] $ that has been proven to be equal the index of a Fredholm operator \cite{Bellissard_1994} and the index of a pair of projections \cite{Avron_Seiler_Simon_1994}. 

\begin{Theorem}
 Given $ P, \, \La_x,\, \La_y $ defined as in theorem \ref{Bott_physics}, it holds:
 \begin{equation} \label{Bott_Chern}
 \B(e^{2\pi i P\Lambda_xP},e^{2\pi i P\Lambda_yP})=2\pi i \T[P\Lambda_xP,P\Lambda_yP]
\end{equation}
\end{Theorem}
\begin{proof}
The following is inspired by an analogous method used to prove the Helton-Howe-Pincus formula.

It has been proven in \cite{Ehrhardt}, see also \cite{Analysis_Toeplitz}, that with $ C, \,D\in \mathcal{L}(\mathcal{H})$ and $ [C,D] \in \mathcal{I}_1 $, $ e^{tC}e^{tD}e^{-t(C+D)}-\I$ is an entire $\mathcal{I}_1 $-valued function. This implies that with $ A:=2\pi i P\Lambda_xP$ and $ B:=2\pi i P\Lambda_yP $, $ e^{tA}e^{tB}e^{-tA}e^{-tB}-\I$ is an entire $\mathcal{I}_1 $-valued function. Moreover equation \eqref{final_bound} implies that $ \forall t \in [0,1] $, $ \|[e^{tA},e^{tB}] \| <2 $. 

Since with $ t \neq 0 $, $ e^{tA} $ is not of the form $ \I + \textrm{trace class} $, then $ \B(e^{tA},e^{tB}) $ will not in general be invariant along the unitary paths $ (e^{tA},e^{tB}) $.  

Let us define $ g(t):= \T \log \left(e^{tA}e^{tB}e^{-tA}e^{-tB} \right)$.
We now map the problem of determining $ \B(e^{A},e^{B})$ into the boundary value of the solution of a differential equation.
\begin{align}
 \frac{dg}{dt}&=\T [\left(Ae^{tA}e^{tB}e^{-tA}e^{-tB}+e^{tA}Be^{tB}e^{-tA}e^{-tB}-e^{tA}e^{tB}Ae^{-tA}e^{-tB}-e^{tA}e^{tB}e^{-tA}Be^{-tB}\right)\nonumber \\ &\hspace{10mm} e^{tB}e^{tA}e^{-tB}e^{-tA} ] \\
 &=\T \left(A+e^{tA}Be^{-tA}-e^{tA}e^{tB}Ae^{-tB}e^{-tA}-e^{tA}e^{tB}e^{-tA}Be^{tA}e^{-tB}e^{-tA}\right) \\
 &=\T \left(A+B-e^{tB}Ae^{-tB}-e^{tB}e^{-tA}Be^{tA}e^{-tB}\right) \\
 &=\T \left(e^{-tB}Ae^{tB}+B-A-e^{-tA}Be^{tA}\right) \\
 &=\T \left(e^{-tB}Ae^{tB} -A \right) + \T \left( B-e^{-tA}Be^{tA} \right) \\
 &=\T \left( \int_0^1 ds e^{-stB}[-tB,A]e^{stB} \right) + \T \left( \int_0^1 ds e^{-stA}[B,-tA]e^{stA} \right) \\
 &=\T \left( [-tB,A] \right) + \T \left( [B,-tA] \right) = 2t \T \left( [A,B] \right) 
\end{align}
Observing that $ g(0)=0$
\begin{align}
 g(t)=\int_0^t g'(s)ds=\int_0^t 2s\T[A,B] ds = t^2\T[A,B] 
\end{align}
This implies
\begin{equation}
 \B(e^{A},e^{B})=\frac{1}{2\pi i}g(1)=\frac{1}{2\pi i} \T[A,B]
\end{equation}
Then 
\begin{equation}
 \B(e^{2\pi i P\Lambda_xP},e^{2\pi i P\Lambda_yP})=2\pi i \T[P\Lambda_xP,P\Lambda_yP].
\end{equation}
\end{proof}

Appendix C of Kitaev's paper \cite{Kitaev_2006} contains a formal proof of the quantization of the transverse Hall conductance, this was a motivation for the authors of \cite{Elgart_Fraas_2021} to identify the trace class properties of the invertible operators $ A $ and $ B $ that ensure $ \det(ABA^{-1}B^{-1}) = 1 $. Applying lemma 2.4 of \cite{Elgart_Fraas_2021} together with the Helton-Howe-Pincus formula \cite{Ehrhardt,Analysis_Toeplitz} we obtain that:
\begin{equation}
 1 = \det \left(e^{2\pi i P\Lambda_xP} e^{2\pi i P\Lambda_yP} e^{-2\pi i P\Lambda_xP} e^{-2\pi i P\Lambda_yP} \right) = e^{ \T[2\pi i P\Lambda_xP,2 \pi i P\Lambda_yP]}
\end{equation}
That implies: $ 2\pi i \T[P\Lambda_xP,P\Lambda_yP] \in \mathbb{Z} $.


\section{Perspectives}
One question that arises naturally is to show a direct connection among the Bott index and the Fredholm index, namely: given two unitary (or invertible) operators $ U,V$, which is the function $ f(U,V) $ such that $ \B(U,V)=\textrm{Index}(f(U,V)) $?
This question has a general interest because if answered affirmatively it would make possible to map the properties of the Fredholm index back to the Bott index. 

Another perspective that this work opens is related to the spectral localizer. We have seen in section \ref{Chern_number} that the Bott index allows to ``bring down'' the operators that appear at the exponents in the unitary operators appearing on the LHS of eq. \eqref{Bott_Chern} in a flexible way. If via homotopic equivalence it would be possible to get the Hamiltonian $ H $ at the exponent and then ``bring it down'' we would obtain an index that start to look like the spectral localizer, moreover the length scale $ L $ associated to the switch functions might allow to effectively restrict the trace on a square $ L \times L $.

\section{Acknowledgements} 

I acknowledge financial support by the UK’s Engineering and Physical Sciences Research Council (grant number EP/R012393/1 Masanes)

\section*{Appendices}
\appendix


\section{Bounding the trace norm of the logarithm} \label{appendix_C_constant}

The proof of the following lemma follows  chapter 8 of Yafaev's book \cite{Yafaev_Mathematical_General}. 
\begin{Lemma} \label{log}  
 Given $ S $ and $ T $ two unitary operators with spectrum not containing $ \{-1\} $, and such that $ \| S-T \|_1 < \infty $. Fixing the branch cut of the logarithm on the real negative axis, it follows that:
 \begin{equation}
 \frac{1}{2 \pi} \| \log S - \log T \|_1 \le C \|S-T\|_1
 \end{equation}
 The constant $ C $ depends on the shape of the loop $ \Gamma $, see the proof.
 \begin{equation} \label{def_C}
  C =  \sup_{w\in \Gamma} \left( \| (w\I-S)^{-1} \| \|(w\I-T)^{-1}\| \right) \frac{1}{4\pi^2} \int_{\Gamma} |\log z| |dz| 
 \end{equation}
 \end{Lemma}

\begin{proof}
 We write $ \log S - \log T $ using the contour integral representation with $ \Gamma $ enclosing the spectrum of $ S $ and $ T $ but not the origin of the complex plane.
\begin{align}
 \log S - \log T & = \frac{1}{2\pi i} \int_{\Gamma} \log z \left( (z\I-S)^{-1} - (z\I-T)^{-1} \right) dz \\
 & = \frac{1}{2\pi i} \int_{\Gamma} \log z  (z\I-S)^{-1}(T-S) (z\I-T)^{-1}  dz 
\end{align}
\begin{align}
\| \log S - \log T \|_1 & \le   \frac{1}{2\pi} \int_{\Gamma} |\log z|  \| (z\I-S)^{-1} \| \|(T-S)\|_1 \|(z-T)^{-1}\|  |dz| \\
 & \le \|(T-S)\|_1 \sup_{w\in \Gamma} \left( \| (w\I-S)^{-1} \| \|(w\I-T)^{-1}\| \right) \frac{1}{2\pi}  \int_{\Gamma} |\log z| |dz| 
\end{align}
\end{proof}


\section{Vanishing of the Bott Index} \label{vanishing_index}

A sufficient condition for the vanishing of the Bott index is given. 
\begin{Lemma}
 The inequality 
\begin{equation} \label{vanish}
 |\textrm{Bott}(U,V)| \le \frac{1}{2\pi} \| \log(UVU^*V^*) \|_1 \le \frac{1}{4} \|[U,V] \|_1
\end{equation}
implies that if $ \|[U,V] \|_1 < 4 $ then $ \textrm{Bott}(U,V)=0 $.
\end{Lemma}

\begin{proof}
Since $ \textrm{Bott}(U,V) $ is an integer, if its modulus  has an upper bound strictly smaller than $ 1 $ then it vanishes. 
 Let us consider the set of eigenvalues of $ UVU^*V^* $, $ \{ e^{i\theta_j} \}$ with $ \theta_j \in (-\pi, \pi) $, it follows that 
 \begin{align}
 &\frac{1}{2\pi} \| \log(UVU^*V^*) \|_1 = \frac{1}{2\pi} \sum_j |\theta_j|  \label{4} \\
 &\le \frac{1}{4} \sum_j |e^{i\theta_j}-1| = \frac{1}{4} \| UVU^*V^* - \mathds{1} \|_1 = \frac{1}{4} \| (UV - VU) U^*V^* \|_1 \le \frac{1}{4} \| [U,V] \|_1 \label{5}
 \end{align}
The singular values of a normal compact operator are the modulus of the eigenvalues; this is used in \eqref{4}. The inequality $ |\theta| \le \frac{\pi}{2} |e^{i\theta}-1| $, $ \theta \in (-\pi,\pi) $, has been used in \eqref{5}.
\end{proof}


\section{Unitary operators that differ by a trace class} \label{trace_class_pert}

\begin{Lemma}
 $ U $ and $ V $ are unitary operators on a Hilbert space $ \mathcal{H} $, it is $ U-V \in \mathcal{I}_1 $ if and only if it exists $ A $ self-adjoint trace class such that $ V=e^{iA}U $.
\end{Lemma}
\begin{proof}
 If $ A $ self-adjoint trace class and $ V=e^{iA}U $ then $ U-V =(\I-e^{iA})U = A \cdot \textrm{Bounded} \in \mathcal{I}_1 $.
If $ U-V \in \mathcal{I}_1 $ then, being $ U-V = (\I-VU^{-1})U $, $ (\I-VU^{-1}) \in \mathcal{I}_1 $. This implies that if there is  essential spectrum of $ (\I-VU^{-1}) $ then it is in $ \{0\} $, therefore the only point where the discrete spectrum can accumulate is $ \{0\} $. Then $ \I-VU^{-1} = \sum_m(1-e^{i\theta_m})P_m $. $ P_m $ is the spectral projection associated to the eigenvalue  $ e^{i\theta_m} $, $ \forall m \hspace{2mm}\textrm{dim} P_m < \infty $, $ \theta_m \in (-\pi,\pi) $.
It is easy to verify that 
\begin{equation} \label{exp_pro}
\sum_me^{i\theta_m}P_m = \sum_me^{i\theta_m P_m} 
\end{equation}
Let us define $ A:= \sum_m \theta_m P_m $, this implies that $ A $ is compact. Let us show that $ A $ is trace class.
\begin{align}
\| A \|_1 = \sum_m |\theta_m| \le \frac{\pi}{2} \sum_m | 1-e^{i\theta_m} | =  \frac{\pi}{2} \| \I-VU^{-1} \|_1 < \infty 
\end{align}
\end{proof}


\section{Proof of theorem \ref{logarithmic_law}} \label{logarithmic_law_proof}
\begin{proof}
To ease the notation $ A $, $ B $ and $ C $ are taken skew self-adjoint.
The strategy of the proof consists of seeing $ \B(e^Ae^B,e^C) $ as the boundary value of the solution of a differential equation. The unitary operator $ e^{tA}e^{tB}e^{tC}e^{-tB}e^{-tA}e^{-tC} $ is an entire function in the variable $ t $ of the type $ \I + \textrm{trace class} $, for a proof see \cite{Ehrhardt}, moreover with $ \|[e^{tA}, e^{tC}] \| <1 $ and $ \|[e^{tB}, e^{tC}] \| <1 $ it follows $ \|[e^{tA}e^{tB}, e^{tC}] \| <2 $ for all $ t \in [0,1] $. This ensure that $  \log \left( e^{tA}e^{tB}e^{tC}e^{-tB}e^{-tA}e^{-tC} \right) $ is well defined and trace class. Let us define $ g(t):= \T \log \left(  e^{tA}e^{tB}e^{tC}e^{-tB}e^{-tA}e^{-tC} \right) $, then $ g(0)=0 $.
\begin{align}
 g'(t)= & \T ( Ae^{tA}e^{tB}e^{tC}e^{-tB}e^{-tA}e^{-tC} + e^{tA}Be^{tB}e^{tC}e^{-tB}e^{-tA}e^{-tC} \\ 
 & \hspace{5mm} + e^{tA}e^{tB}Ce^{tC}e^{-tB}e^{-tA}e^{-tC} - e^{tA}e^{tB}e^{tC}Be^{-tB}e^{-tA}e^{-tC} \\
 & \hspace{5mm}- e^{tA}e^{tB}e^{tC}e^{-tB}Ae^{-tA}e^{-tC} - e^{tA}e^{tB}e^{tC}e^{-tB}e^{-tA}Ce^{-tC} ) \\
 & \hspace{5mm} \cdot e^{tC}e^{tA}e^{tB}e^{-tC}e^{-tB}e^{-tA}\\
 =& \T ( e^{-tB}Ae^{tB} + B + C - e^{tC}Be^{-tB} - \\
 & \hspace{5mm} e^{tC}e^{-tB}Ae^{tB}e^{-tC} - e^{tC}e^{-tB}e^{-tA}Ce^{tA}e^{tB}e^{-tC} ) \\
 =& \T \left( e^{-tC}e^{-tB}Ae^{tB}e^{tC} + e^{-tC}B e^{tC} + C - B - e^{-tB}Ae^{tB} - e^{-tB}e^{-tA}Ce^{tA}e^{tB} \right) \\
 =& \T \left( e^{-tC}e^{-tB}Ae^{tB}e^{tC} - e^{-tB}Ae^{tB} \right) + \T \left( e^{-tC}B e^{tC} - B \right) \\
 &+ \T \left(C - e^{-tB}e^{-tA}Ce^{tA}e^{tB} \right) \\
 =& \T \left( \int_0^1 ds e^{-stC}[-Ct,e^{-tB}Ae^{tB}]e^{stC} \right) +  \T \left( \int_0^1 ds e^{-stC}[-Ct,B]e^{stC} \right)  \\
 & + \T \left(C -  e^{-tA}Ce^{tA} + e^{-tA}Ce^{tA} - e^{-tB}e^{-tA}Ce^{tA}e^{tB} \right) \\ 
 =& -t \T \left( [C,e^{-tB}Ae^{tB}] \right) - t \T \left( [C,B] \right) + \T \left( \int_0^1 ds e^{-stA}[C,-tA]e^{stA} \right) \\
 & \hspace{5mm} + \T \left( \int_0^1 ds e^{-stB}[tB,e^{-tA}Ce^{tA}]e^{stB} \right) \\
 =& -t \T \left( [e^{tB}Ce^{-tB},A] \right) -t \T \left( [C,B] \right) -t \T ([C,A]) + \T \left( [tB,e^{-tA}Ce^{tA}] \right) \\
 =& 2t \T ([A,C]) + 2t \T ([B,C]) \\
 &g(t)-g(0)=\int_0^t ds g'(s) = t^2 \T ([A,C]) + t^2 \T ([B,C]) 
\end{align}
With the addition constraints $ (e^A-\I)(e^C-\I) \in \mathcal{I}_1 $, $ (e^C-\I)(e^A-\I) \in \mathcal{I}_1 $, $ (e^B-\I)(e^C-\I) \in \mathcal{I}_1 $, $ (e^C-\I)(e^B-\I) \in \mathcal{I}_1 $, ensuring that the Bott index is an integer, it follows:  
\begin{align}
\B(e^Ae^B,e^C) = \frac{g(1)}{2 \pi i} = \frac{1}{2 \pi i} ( [A,C]) + [B,C]) = \B(e^A,e^C) + \B(e^B,e^C)
\end{align}
We note that $ (e^A-\I)(e^C-\I) \in \mathcal{I}_1 $, $ (e^C-\I)(e^A-\I) \in \mathcal{I}_1 $, $ (e^B-\I)(e^C-\I) \in \mathcal{I}_1 $, $ (e^C-\I)(e^B-\I) \in \mathcal{I}_1 $ implies $ (e^Ae^B-\I)(e^C-\I) \in \mathcal{I}_1 $, $ (e^C-\I)(e^Ae^B-\I) \in \mathcal{I}_1 $    in fact:
\begin{align}
 &(e^Ae^B-\I)(e^C-\I) = e^A(e^B-e^{-A})(e^C-\I) = e^A(e^B-\I -e^{-A}+\I)(e^C-\I) \\
 & = e^A(e^B-\I)(e^C-\I) - e^A(e^{-A}-\I)(e^C-\I) \\ & =  e^A(e^B-\I)(e^C-\I) - (\I-e^A)(e^C-\I) \in \mathcal{I}_1 
\end{align}
This guarantees that $ \B(e^Ae^B,e^C) $ is an integer.
\end{proof}

\bibliography{biblio}

\end{document}